\documentclass[%
 reprint,
 amsmath,amssymb,
 aps,
floatfix,
nobalancelastpage, 
]{revtex4-2}

\usepackage{dcolumn}
\usepackage{amsfonts}
\usepackage{amsmath}
\usepackage{amssymb}
\usepackage{amsthm}
\usepackage{mathtools}
\usepackage{physics}
\usepackage{dsfont}
\usepackage[dvipsnames]{xcolor}
\usepackage{enumitem}
\usepackage{subcaption}
\usepackage{tikz}
\usepackage{microtype}
\usepackage{hyperref}
\usepackage{multirow}

\theoremstyle{plain}
\newtheorem{lemma}{Lemma}[section]
\theoremstyle{definition}
\newtheorem{definition}[lemma]{Definition}

\newcommand{\df}{\mathrm{d}}

\usepackage{natbib}
\begin{document}
\preprint{APS/123-QED}

\title{Inflated Graph States Refuting Communication-Assisted LHV Models}

\author{Uta Isabella Meyer}
\email{uta-isabella.meyer@lip6.fr}
\author{Frédéric Grosshans}
\email{frederic.grosshans@lip6.fr}
\author{Damian Markham}
\email{damian.markham@lip6.fr}
\affiliation{Sorbonne Université, CNRS, LIP6, F-75005 Paris, France}

\date{\today}

\begin{abstract}
Standard Bell inequalities hold when distant parties are not allowed to communicate. Barrett et al.\@ found correlations from Pauli measurements on certain network graphs refute a local hidden variable (LHV) description even allowing some communication along the graph.
This has recently found applications in proving separation between classical and quantum computing, in terms of shallow circuits, and distributed computing. 
The correlations presented by Barrett et al.\@ can be understood as coming from an extension of three party GHZ state correlations which can be embedded on a graph state.
In this work, we propose systematic extensions of any graph state, which we dub \textit{inflated graph states} such that they exhibit correlations which refute any communication assisted LHV model.
We further show the smallest possible such example, with a 7-qubit linear graph state, as well as specially crafted smaller examples with 5 and 4 qubits.
The latter is the smallest possible violation using binary inputs and outputs. 

\end{abstract}
\maketitle

\hyphenation{sub--mea-sure-ment}
\section{Introduction}
A simple undirected graph is a set of vertices and edges connecting the vertices. Graph states are defined in one to one correspondence with a simple undirected graph, where the vertices represent qubits, and each edge corresponds to a preparation entanglement operation.
In quantum information, graph states form an important class of multipartite entangled states \cite{hein2004multiparty} including Bell states, GHZ states \cite{Greenberger1989}, stabiliser code states and cluster states. Their application ranges from universal resources for quantum computation \cite{raussendorf2001one}, to error correction and fault tolerance \cite{schlingemann2001quantum,campbell2017roads}, quantum metrology \cite{shettell2020graph} and to quantum network protocols such as secret sharing \cite{markham2008graph} and anonymous communication \cite{christandl2005quantum}.

Non-locality is an increasingly recognised resource for quantum information \cite{brunner2014bell}. The standard setting for non-locality is of separated laboratories which are prohibited from communicating, where we look at the correlations that come out of local measurements. If they cannot be reproduced by a local hidden  variable models, or equivalently with shared classical randomness, we say they are non-local, and they give rise to quantum advantages in non-local games and communication complexity and are behind device independent security  \cite{brunner2014bell}. Any non-trivially connected graph state can demonstrate non-locality \cite{guhne2005bell}. 
In particular no local hidden variable (LHV) model can describe all Pauli measurements on graph states with more than two connected vertices \cite{scarani2005nonlocality}.

Barrett et al.\@ \cite{barrett2007modeling} showed that certain graph states still display non-locality even when the LHV model is assisted by classical communication along the graph's edges. Their construction essentially takes the three party GHSZ paradox \cite{greenberger1990bell}, which corresponds locally to a triangle graph state where different Pauli measurements are performed on each corner, and extends it by adding vertices along the edges and performing fixed measurements on these additional vertices. The resulting correlations not only violate an LHV description, but also an LHV description extended to allow communication along the edges of the graph up to a fixed distance.

As well as its fundamental interest, this extension of the standard setting for non-locality has given rise to several new applications. It is behind the only assumption free proof of the separation between quantum and classical computational power in the context of shallow circuits \cite{bravyi2018quantum}. It has further been used to show quantum advantage in distributed computing \cite{gall2018quantum} and to certify randomness replacing the assumptions of no-communication by assumptions on circuit depth of an adversary \cite{coudron2018trading}.


In this work we extend the class of graph states that exhibit non-locality when allowing distance-bounded classical communication. 
We follow the model for networks and allowed communication as presented by Barrett et al.\ \cite{barrett2007modeling}.
For bounding classical communication we assume that the vertices can exchange classical information if they are connected by a bounded (yet not vanishing) number of edges. 
Note that this is a reasonable assumption as, given a graph state, vertices sharing an edge must have interacted in the past (at least indirectly). 
In particular we consider classical communication restricted to a fixed bounded distance $d$ over the graph edges. We denote an LHV model assisted by such classical communication as $d$-LHV$^\ast$.
We construct graph states and formulate GHSZ-like paradoxes \cite{greenberger1990bell} that impose an all-or-nothing constraint on any $d$-LHV$^\ast$ model. 
Using the formulation by Mermin \cite{mermin1990simple}, this leads to a Bell inequality that a given graph state violates.
We also explore examples for small graph states.
We show that the smallest graph state that refutes a $d$-LHV$^\ast$ model with a GHSZ-like paradox using Pauli measurements is a 5-qubit cycle.
We also provide a Bell inequality for the smallest possible graph state of four qubits that violates any $d$-LHV$^\ast$ model
with binary inputs and outputs. It uses Clifford operations instead of Pauli operators.

%
The article is structured as follows. 
We begin in section \ref{sec:graph} with an introduction of graph states and some definitions for our construction. In section \ref{sec:non} we present general conditions for sets of measurements on graph states whose correlations cannot be mimicked by any LHV-model, and then any $d$-LHV$^\ast$-model. In section \ref{sec: inflated graph states} we present the inflated graph states, and some definitions and sets of measurements used for our construction. Our main result appears in section \ref{sec: main result}. We define sets of measurements for any inflated graph state, and prove that they satisfy the conditions set out in section \ref{sec:non}, hence proving that they cannot be mimicked by any $d$-LHV$^\ast$-model. In section \ref{sec: bell inquality} we present the simple Bell inequalities, and in section \ref{sec:example} we explore small graph state examples. We close with discussions in section \ref{sec:conc}.

\section{Graph States}
\label{sec:graph}
We denote a graph $G = (V,E)$ with a set of vertices $V$ and edges $E \subset (V\times V)$. 
When each vertex represents a qubit, the associated graph state $\vert G \rangle$ is a quantum state on $n = \vert V \vert$ qubits that uniquely corresponds to the graph $G$. 
The graph state is defined as the unique eigenvector with eigenvalue $1$ to the generator elements
\begin{equation}  g_v = X_v \bigotimes_{(u,v)\in E} Z_u\,, \label{eq:vertexstabelement}\end{equation} 
for all $v \in V$ and with the Pauli operators $\sigma_v=\lbrace \mathds{1}_v, X_v,Y_v,Z_v \rbrace$. 
The generator elements commute pairwise and generate the stabiliser group $\mathcal{S} = \langle \lbrace g_v\,, v\in V\rbrace \rangle$. As such any product is a stabiliser element with a decomposition 
\begin{equation} S = \prod_{u\in U\subset V} g_u \,.\label{eq:stab-decomposition}\end{equation} 
Applying Eq.~\eqref{eq:vertexstabelement} we can write any stabiliser element as a product of Pauli operators and a sign factor \begin{equation} S = \chi^{\phantom{\prime}} \bigotimes_{v\in V} \sigma_v \,.\label{eq:stab-decomposition2}\end{equation}
We will now specify a map between both representations, Eq.~\eqref{eq:stab-decomposition} and Eq.~\eqref{eq:stab-decomposition2}, that is local to the vertices and its nearest neighbours.
Starting from any stabiliser element $S$ as a product of generator elements with an index set $U \subset V$, we determine\begin{equation} \delta_{u} = \begin{cases} 0\,,~ u \notin U\,,\\    1 \,,~ u \in U\,,    \end{cases}\end{equation} and $t_u = \sum_{(u,v)\in E} \delta_{v} \bmod 4$ for every vertex $u \in V$. 
The acquired three bits of information suffice to uniquely assign a Pauli operator with a sign factor to every vertex by using 
\begin{equation}
\chi^{\phantom{\prime}}_u \sigma_u = \mathrm{i}_{\phantom{u}}^{\delta_u t_u} X_u^{\delta_u} Z_u^{t_u} \,.\label{eq:rules}
\end{equation}
Then, the sign factor in Eq.~\eqref{eq:stab-decomposition2} is $\chi^{\phantom{\prime}} = \prod_{v\in V} \chi^{\phantom{\prime}}_v$.
Conversely, applying Eq.~\eqref{eq:rules} can be used to determine whether any Pauli product is a stabiliser element and provide the decomposition into generator elements.\\

When using graph states, we are very often interested in performing Pauli measurements, whose operators are tensor products of Pauli operators. On a given graph state Pauli measurements have deterministic outcome $\pm 1$ if they are proportional to a stabilizer element as in Eq.~\eqref{eq:stab-decomposition2}, otherwise the outcome is random with expectation value $0$. 
When measuring a Pauli product, the outcome can be regarded as a product of the local Pauli operator's outcome.
Checking a stabiliser condition, for example, corresponds to checking that the product of all the involved Pauli operators gives the correct sign (according to (\ref{eq:stab-decomposition2})).
Barrett et al.\cite{barrett2007modeling} provide a $1$-LHV$^\ast$ model that correctly predicts the outcome any Pauli measurement performed on a graph state. However, they also show that by ignoring certain vertices output they abrogate any $d$-LHV$^\ast$ description. This motivates the following definition. 
%
\begin{definition}[Submeasurement]
Given a Pauli measurement $M = \bigotimes_{v\in V} \sigma_v$ on the graph state, a submeasurement of $M$ is defined by a Pauli product $C = \bigotimes_{v\in V} \sigma^\prime_v$ such that $\sigma^\prime_v = \sigma_v$ or $\sigma^\prime_v = \mathds{1}_v$ for all $v \in V$.
\end{definition}
When considering the correlations, all the vertices where the submeasurement's operator differs from the Pauli measurement's are those whose outcomes are disregarded.

\section{Non-locality in Graph States}\label{sec:non}

In this section we will review non-locality of graph states, and present an alternative approach to that of \cite{guhne2005bell}, which we then extend to show non-locality persists when limited classical communication is allowed between vertices.

Graph states exhibit non-local properties when considering a vertex as a local unit. That is, they contradict the prediction of a measurement's outcome made by a local hidden variable (LHV) model \cite{guhne2005bell}. 
We formulate LHV models as follows.
We denote a hidden variable $\lambda$ with a probability distribution $\rho(\lambda)$. 
To each $\lambda$, the LHV assigns deterministic variables $h(O,\lambda)$ to local observables $O$. 
A measurement's outcome is then a probabilistic mixture of the deterministic variable over the distribution of the hidden variable $\lambda$. As the variables' values are the possible observables' outcome, it is $h(\sigma_v,\lambda)= \pm 1$ for Pauli operators $\sigma_v = X_v,Y_v,Z_v$ and $h(\mathds{1}_v,\lambda)= 1$. Given a Pauli measurement $M = \bigotimes_{v\in V} \sigma_v$, the LHV model predicts the outcome as $\langle M \rangle_{\mathrm{LHV}} = \int \rho \left(\lambda \right) \prod_{v \in V} h \left(\sigma_u, \lambda \right) \, \df \lambda$. In order to show a contradiction with any LHV model, it suffices to show that all arbitrary choices of deterministic variables $h$ sustain the contradiction, thus a probabilistic mixture does as well. We can, therefore, drop the notion of a hidden variable $\lambda$ but consider all possible deterministic variables' outcomes instead.

We briefly outline the arguments for non-locality of graph states, which will be the starting point of our constructions. A set of Pauli measurements $\lbrace M_k \rbrace_k$ on a fixed graph state contradicts any LHV model if it fulfills the three properties. First,
\begin{equation}
    \vert \lbrace k; (M_k)_v = \sigma_v \rbrace \vert \bmod 2  =0 \,, \forall \sigma_v \neq \mathds{1}_v \,,\label{ali:cond-3}
\end{equation}
for all $v \in V$ and the set cardinality $|\cdot|$, which states that all non-trivial Pauli operators must occur an even number of times on very vertex. Second and third,
\begin{align} \phantom{\prod_{k}}M_k &= \chi^{\phantom{\prime}}_k S_k \,, \label{ali:cond-1} \\\prod_{k} M_k &= -\mathds{1}^{\otimes^V} \,, \label{ali:cond-2} \end{align} 
with a stabiliser elements $S_k$ and sign factors $\chi^{\phantom{\prime}}_k = \pm 1$.
Such a set exists for all graph states with at least three connected vertices, as shown in \cite{guhne2005bell}.

Equation \eqref{ali:cond-3} states that on every vertex, the number of every non-trivial Pauli operators the set of Pauli measurements $\lbrace M_k \rbrace_k$ is even. 
For any LHV this implies \begin{align*}
    \prod_{k} \langle M_k \rangle_{\mathrm{LHV}} 
    &= \prod_{k} \int \rho(\lambda)\prod_{v\in V} h((M_k)_v,\lambda) \,\df\lambda
    \\&= \int \rho(\lambda) \prod_{k^\prime} \prod_{v\in V} h^2((M_{k^\prime})_v,\lambda)  \,\df\lambda
    \\&= \int \rho(\lambda) \prod_{v\in V} (\pm 1)^2 \,\df\lambda = 1\,,
\end{align*} where we used $h((M_k)_v=\mathds{1}_v,\lambda)=1$.
On the other hand, Eq.~\eqref{ali:cond-1} implies that quantum mechanics predicts every measurement's outcome to be deterministically $\chi^{\phantom{\prime}}_k = \pm 1$. Furthermore, from Eq.~\eqref{ali:cond-2} the overall product of measurements' results is then deterministically $\prod_{k} \langle M_k \rangle_{\mathrm{QM}} = -1$. The contradiction implies that no LHV model can reproduce the results predicted by quantum mechanics.\\


We now move to consider classical models where distance-bounded communication is permitted.

\begin{definition}
    Given a graph $G$ and non-negative integer $d$, a graph's vertex serves as a local unit and distance is measured along the edges. Then, a communication-assisted LHV model $d$-LHV$^\ast$ is a hidden variable model where each hidden variable depends on a local input and inputs up to distance $d$, which accounts for vertices broadcasting information about their inputs.
\end{definition}

We define slightly adapted conditions on the measurements, which will similarly lead to contradictions. Considering Pauli measurements, $d$-LHV$^\ast$ model specifies randomly predetermined variables conditioned on the Pauli operators local to a vertex and its neighbours up to distance $d$ edges apart. For every vertex $u$ and $M_k$, we define an excerpt of local measurements around the vertex $v$ with maximum distance $d$ along the edges
\begin{equation}
   (\mathbf{M}_k)^d_v := \bigotimes_{\substack{u\\\lvert u- v \rvert\leq d}}(M_k)_u \,,
\end{equation}
which is a tensor product of Pauli operators $\sigma^d_v = \bigotimes_{u,\,\vert u- v \vert \leq d} \sigma_u$ around the vertex $v$.

In order to obtain a contradiction between the predictions of quantum mechanics and those allowed by such a $d$-LHV$^\ast$, we consider now pairs of measurements $(M_k, C_k)$, where the $M_k$ will be the measurements that the vertices are asked to perform, but the $C_k$, which are submeasurements of the $M_k$, are the measurements that lead to contradictions similar to what we have above. Recall that the submeasurement's Pauli operators are 
either the identity, or identical to the Pauli measurement. 
The role of the $M_k$ is then to hide from any $d$-LHV$^\ast$ information that identifies fully the correlations being checked.

For a communication-assisted version of Eq.\ \eqref{ali:cond-3}, the number of Pauli products $(\mathbf{M}_k)^d_v$ must be even in the set of measurements conditioned on $(C_k)_v \neq \mathds{1}_v$, that is
\begin{equation}
    \vert \lbrace k\,; (\mathbf{M}_k)^d_v = \mathbf{\sigma}^d_v \wedge (C_{k})_v \neq \mathds{1}_v \rbrace  \vert  \bmod 2 =0 \,,~\forall \sigma^d_v\,, \label{eq:infcond-3}
\end{equation}
for all vertices $u \in V$.
Note that for $d=0$ the above condition is equivalent to Eq.~\eqref{ali:cond-3} for a given vertex. Again, for $d>0$, Eq.\,\eqref{eq:infcond-3} requires not only that every non-trivial Pauli operator occurs in pairs in the set of submeasurements but that, at the same time, the corresponding Pauli measurements are equal on the neighbour vertices up to distance $d$. 

We will now see contradictions between any $d$-LHV$^\ast$ model and sets of measurements pairs $(M_k, C_k)$ that satisfies \eqref{eq:infcond-3} for all vertices, and
\begin{align} \phantom{\prod_{k}}C_k &= \chi^{\phantom{\prime}}_k S_k \,, \label{ali:infcond-1} \\\prod_{k} C_k &= -\mathds{1}^{\otimes^V} \,. \label{ali:infcond-2}\end{align}

\begin{lemma} \label{lem:comp} If a graph state has sets of measurements and submeasurement pairs $\lbrace(M_k, C_k)\rbrace_k$, such that the $\lbrace C_k\rbrace_k $ satisfy \eqref{ali:infcond-1},\eqref{ali:infcond-2} and the pairs satisfy Eq.\,\eqref{eq:infcond-3} for all vertices, then their statistics cannot be mimicked by any $d$-LHV$^\ast$ model. \end{lemma}

\begin{proof}
When one performs Pauli measurements on a graph state, any $d$-LHV$^\ast$ model randomly assigns predetermined variables $h$ for each vertex depending on the Pauli operators on the vertex and its neighbours up to distance $d$. In the submeasurements, we discard certain local outcomes. 
If they fulfill Eq.~\eqref{eq:infcond-3} on every vertex, 
all variables conditioned on different Pauli operators occur an even number of times among the set of measurements and $h((C_k)_u,=\mathds{1}_u \lbrace (M_k)_v \rbrace_{0<\vert u-v \vert \leq d},\lambda)=1$. For any $d$-LHV$^\ast$ model we evaluate the product of all submeasurements' outcomes
\begin{align*}
    &\phantom{=}\prod_{k} \langle C_k \rangle_{d-\mathrm{LHV}^\ast}
    \\ &= \prod_{k} \int \rho(\lambda) \prod_{u\in V}  h((C_k)_u, \lbrace (M_k)_v \rbrace_{0<\vert u-v \vert \leq d},\lambda) \,\df \lambda \allowbreak
    \\ &= \int \rho(\lambda) \prod_{u\in V} \prod_{k} h((\mathbf{M_k})^d_u,\lambda)\, \df \lambda
    \\ &= \int \rho(\lambda) \, \df \lambda = 1 \,.
\end{align*}
Contrarily, Eq.\ \eqref{ali:infcond-1} implies that quantum mechanics predicts every submeasurement's outcome to be deterministically $\chi^{\phantom{\prime}}_k = \pm 1$ . Furthermore, from Eq.~\eqref{ali:infcond-2} the overall product of submeasurements' results is then deterministically $\prod_{k} \langle C_k \rangle_{\mathrm{QM}} = -1$. The contradiction implies that no $d$-LHV$^\ast$ model can reproduce the results predicted by quantum mechanics. \end{proof}

In Tab.\,\ref{tab:ghz}, we see how the cycle graph in Fig.\,\ref{fig:triangle9}, viewed as a triangle with GHSZ like correlations at the corners gives such a contradiction, which is just a different presentation of the results of \cite{barrett2007modeling}. In the following we will extend this idea to arbitrary graphs.
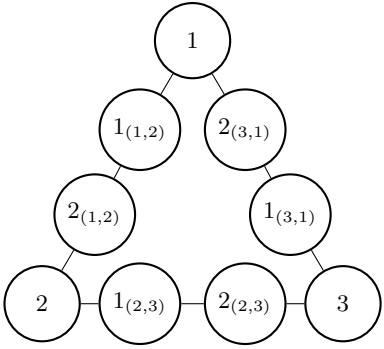
\begin{figure}
    \centering
    \begin{tikzpicture}
\begin{scope}[every node/.style={circle,thick,draw}]

    \node[minimum size=1cm] (1) at (2,3.5) {$1$};
    \node[minimum size=1cm] (2) at (0,0) {$2$};
    \node[minimum size=1cm] (3) at (4,0) {$3$};

    \node (121) at (0.7,1.185) {$2_{(1,2)}$};
    \node (122) at (1.3,2.315) {$1_{(1,2)}$};

    \node (131) at (2.7,2.315) {$2_{(3,1)}$};
    \node (132) at (3.3,1.185) {$1_{(3,1)}$};
    
    \node (231) at (1.3,0) {$1_{(2,3)}$};
    \node (232) at (2.7,0) {$2_{(2,3)}$};

    \draw[-] (2) -- (121); 
    \draw[-] (121) -- (122);
    \draw[-] (122) -- (1);

    \draw[-] (1) -- (131);
    \draw[-] (131) -- (132);
    \draw[-] (132) -- (3);

    \draw[-] (2) -- (231);
    \draw[-] (231) -- (232);
    \draw[-] (232) -- (3);

    \end{scope}
    \end{tikzpicture}

    \caption{Circular graph seen as a $(d=1)$-inflated triangle graph whose graph state is locally equivalent to the $3$-GHZ state. The nodes contain the vertices' label. We provide a set of Pauli measurements and submeasurements in Tab.\,\ref{tab:ghz} that no $(d=1)$-LHV$^\ast$ model can correctly predict.}
    \label{fig:triangle9}
\end{figure}
\begin{table}
\centering
\begin{tabular}{clllllllllll}
$M^\prime_1=$ &  $X_1$ & \textcolor{gray}{$X_{1_{(1,2)}}$} & \hspace{-0.1cm}$X^{\prime}_{2_{(1,2)}}$ & $Z_2$ & \textcolor{gray}{$X_{1_{(2,3)}}$} & \hspace{-0.1cm}\textcolor{gray}{$X_{2_{(2,3)}}$} & $Z_3$ & $X^{\prime}_{1_{(3,1)}}$ & \hspace{-0.1cm}\textcolor{gray}{$X_{2_{(3,1)}}$}  \\
$M^\prime_2=$ & $Z_1$ & $X^{\prime}_{1_{(1,2)}}$ & \hspace{-0.1cm}\textcolor{gray}{$X_{2_{(1,2)}}$} & $X_2$ & \textcolor{gray}{$X_{1_{(2,3)}}$} & \hspace{-0.1cm}$X^\prime_{2_{(2,3)}}$ & $Z_3$ & \textcolor{gray}{$X_{1_{(3,1)}}$} & \hspace{-0.1cm}\textcolor{gray}{$X_{2_{(3,1)}}$}  \\
$M^\prime_3=$ & $Z_1$ & \textcolor{gray}{$X_{1_{(1,2)}}$} & \hspace{-0.1cm}\textcolor{gray}{$X_{2_{(1,2)}}$} & $Z_2$ & $X^{\prime}_{1_{(2,3)}}$ & \hspace{-0.1cm}\textcolor{gray}{$X_{2_{(2,3)}}$} & $X_3$ & \textcolor{gray}{$X_{1_{(3,1)}}$} & \hspace{-0.1cm}$X^{\prime}_{2_{(3,1)}}$  \\
$M^\prime_4=$ & $X_1$ & $X_{1_{(1,2)}}$ & \hspace{-0.1cm}$X_{2_{(1,2)}}$ & $X_2$ & $X_{1_{(2,3)}}$ & \hspace{-0.1cm}$X_{2_{(2,3)}}$ & $X_3$ & $X_{1_{(3,1)}}$ & \hspace{-0.1cm}$X_{2_{(3,1)}}$  \\
$\tilde{M}^\prime_1=$ & \textcolor{gray}{$X_1$} & $X_{1_{(1,2)}}$ & \hspace{-0.1cm}\textcolor{gray}{$X_{2_{(1,2)}}$} & $Y_2$ & $X^{\prime \prime}_{1_{(2,3)}}$ & \hspace{-0.1cm}$X^{\prime \prime}_{2_{(2,3)}}$ & $Y_3$ & \textcolor{gray}{$X_{1_{(3,1)}}$} & \hspace{-0.1cm}$X_{2_{(3,1)}}$ \\
$\tilde{M}^\prime_2=$ & \textcolor{gray}{$Z_1$} & $X^{\prime}_{1_{(1,2)}}$ & \hspace{-0.1cm}\textcolor{gray}{$X_{2_{(1,2)}}$} & $Y_2$ & $X^{\prime \prime}_{1_{(2,3)}}$ & \hspace{-0.1cm}$X^{\prime \prime}_{2_{(2,3)}}$ & $Y_3$ & \textcolor{gray}{$X_{1_{(3,1)}}$} & \hspace{-0.1cm}$X^\prime_{2_{(3,1)}}$ \\
$\tilde{M}^\prime_3=$ & $Y_1$ & \textcolor{gray}{$X_{1_{(1,2)}}$} & \hspace{-0.1cm}$X_{2_{(1,2)}}$ & \textcolor{gray}{$X_2$} & $X_{1_{(2,3)}}$ & \hspace{-0.1cm}\textcolor{gray}{$X_{2_{(2,3)}}$} & $Y_3$ & $X^{\prime \prime}_{1_{(3,1)}}$ & \hspace{-0.1cm}$X^{\prime \prime}_{2_{(3,1)}}$ \\
$\tilde{M}^\prime_4=$ & $Y_1$ & \textcolor{gray}{$X_{1_{(1,2)}}$} & \hspace{-0.1cm}$X^{\prime}_{2_{(1,2)}}$ & \textcolor{gray}{$Z_2$} & $X^\prime_{1_{(2,3)}}$ & \hspace{-0.1cm}\textcolor{gray}{$X_{2_{(2,3)}}$} & $Y_3$ & $X^{\prime \prime}_{1_{(3,1)}}$ & \hspace{-0.1cm}$X^{\prime \prime}_{2_{(3,1)}}$ \\
$\tilde{M}^\prime_5=$ & $Y_1$ & $X^{\prime \prime}_{1_{(1,2)}}$ & \hspace{-0.1cm}$X^{\prime \prime}_{2_{(1,2)}}$ & $Y_2$ & \textcolor{gray}{$X_{1_{(2,3)}}$} & \hspace{-0.1cm}$X_{2_{(2,3)}}$ & \textcolor{gray}{$X_3$} & $X_{1_{(3,1)}}$ & \hspace{-0.1cm}\textcolor{gray}{$X_{2_{(3,1)}}$} \\
$\tilde{M}^\prime_6=$ & $Y_1$ & $X^{\prime \prime}_{1_{(1,2)}}$ & \hspace{-0.1cm}$X^{\prime \prime}_{2_{(1,2)}}$ & $Y_2$ & \textcolor{gray}{$X_{1_{(2,3)}}$} & \hspace{-0.1cm}$X^{\prime}_{2_{(2,3)}}$ & \textcolor{gray}{$Z_3$} & $X^\prime_{1_{(3,1)}}$ & \hspace{-0.1cm}\textcolor{gray}{$X_{2_{(3,1)}}$} 
\end{tabular}

\caption{Set of Pauli measurements (black and grey text) and with submeasurements (grey operators are treated as $\mathds{1}$) on the graph in Fig.\,\ref{fig:triangle9} that contradicts any $d=1$-LHV$^\ast$ model. In order to account for a round of classical communication among nearest neighbors, the default is that the nearest neighbours are measuring Pauli $X$ and we prime an operator when a nearest neighbor measures Pauli $Z$ and double prime when Pauli $Y$. In the first four measurements, the vertices $1,2,3$ mimic a GHSZ-like paradox on a triangular graph state. The remaining measurements are hiding the classical information acquired by the other vertices from the main three vertices. Altogether the submeasurements fulfill Eq.\,\eqref{ali:infcond-1} and Eq.\,\eqref{ali:infcond-2}, and the pairs fulfill Eq.\,\eqref{eq:infcond-3} on all vertices.}
\label{tab:ghz}
\end{table}

\section{Inflated Graph States} \label{sec: inflated graph states}
In this section we define a method to generate examples from any graph states, which we call inflated graph states, and the measurements that will be used in our construction.

For any graph with at least three connected vertices we construct an inflated graph by substituting every edge by chain of $2d$ vertices bridging the edge. Starting from a graph $G = (V,E)$, the inflated graph $G^\prime = (V^\prime, E^\prime)$ has the same vertices as the old graph, completed by $2d$ vertices $r_{(u,v)}$ for every edge $(u,v) \in E$ and $r \in [1,2d]$. The new edges are $(u,1_{(u,v)}), ((2d)_{(u,v)},v) \in E^\prime $ and $(r_{(u,v)},(r+1)_{(u,v)}) \in E^\prime$ for all $r \in [1,2d-1]$ and $(u,v)\in E$. Note that even though the notation of the edges is symmetric ($(u,v)=(v,u)$), we break it by counting the vertices starting from the first vertex to the second as above. We refer to the original vertices as \textit{power} vertices and to the newly added vertices as \textit{chain} vertices. We denote Pauli measurements and stabiliser elements with the primed letter used for the original graph state. Figure \ref{fig:inflated1} shows an exemplary inflated graph.
\begin{definition}[Inflated Measurements] \label{def:powerm} Given a Pauli measurement $M$ on the original graph state, the \textit{Inflated} measurement $M^\prime$ is a Pauli measurement on the inflated graph state. On the power vertices $u$ it measures $(M^\prime)_u = (M)_u$, and $(M^\prime)_{w} = X_{w}$ on all chain vertices $w =r_{(u,v)}$ with $r \in [1,2d]$ and $(u,v) \in E$.\end{definition}
\begin{definition}[Inflated generator element]
The inflated generator element $f_u$ is a stabiliser element of the inflated graph state with
\begin{equation}
f_u = g^\prime_u  
\prod_{\substack{(u,v)\in E,\\s=1,\dots,d}} 
g^\prime_{(2s)_{(u,v)}} \,.\label{eq:infvstab}
\end{equation}
\label{def:infvstab}
\end{definition}
Writing the inflated generator elements in terms of Pauli operators $f_u = X_u  \bigotimes_{v,(u,v)\in E} Z_v \,\bigotimes^d_{s=1} X_{(2s)_{(u,v)}}$, we see that they mimic the generator elements of the original graph state $g_v = X_v \bigotimes_{(u,v) \in E} Z_u$ on the power vertices and measure $X$ on every second chain vertex.
\begin{definition}[Inflated Stabiliser Element]\label{def:powers} Given a stabilizer element $S = \prod_{u \in U} g_u$ with $U \subset V$ on the original graph state, we define the \textit{Inflated Stabiliser} element 
\begin{equation}
    S^\prime = \prod_{u \in U} f_u \,,
    \label{eq:inheritstab}
\end{equation}
on the inflated graph state with the inflated generator elements $f_u$ from Def.\,\ref{def:infvstab}.
\end{definition}

\begin{definition}[Shell Stabiliser Element] \label{def:shell} 
For a vertex $u \in V$, with at least two nearest neighbours $v_1,v_2 \in V$, we define the Shell stabiliser element \begin{equation}
    \tilde{S} = f_{v_1} f_{v_2}
\end{equation} 
on the inflated graph state with the inflated generator elements $f_u$ from Def.\,\ref{def:infvstab}.
\end{definition}

\begin{definition}[Decoy Measurements] \label{def:decoy}
For a vertex $u \in V$ with two nearest neighbours $v_1,v_2 \in V$, and two different Pauli operators $\sigma^{\phantom{\prime}}_u \neq \sigma^{\prime}_u$, we define a pair of Pauli measurements $(\tilde{M}_1,\tilde{M_2})$ called \textit{decoy measurements} on the inflated graph state. On the power vertex $u$ they measure \[(\tilde{M}_1)_u = \sigma^{\phantom{\prime}}_u \,,\,(\tilde{M}_2)_u = \sigma^{\prime}_u \,. \] 
On the chain vertices $w$ linking $u$ to $v_1$ and $v_2$,
\[(\tilde{M}_1)_{w} = (\tilde{M}_2)_{w} = X_{w} \,, \] 
with $w = r_{(u,v_{1(2)})}$ for all $r\in [1,2d]$, and \[(\tilde{M}_1)_v = (\tilde{M}_2)_v = (\tilde{S})_v\,,\] on all other vertices $v \notin \{u\}\cup\{ r_{(u,v)}\}_{v\in\{v_1,v_2\},r \in [1,2d]}$ of the inflated graph for the Shell stabiliser element $\tilde{S}$ defined by the same vertices $u,v_1,v_2$.
\end{definition}

\begin{lemma} Given the same vertices $u,v_1,v_2$, the Shell stabiliser element $\tilde{C} = \tilde{S}$ from Def.\,\ref{def:shell} is a submeasurement of both decoy measurements $(\tilde{M}_1,\tilde{M_2})$ from Def.\,\ref{def:decoy}.\end{lemma}

\begin{proof} We determine the local Pauli operators and the sign factor of the Shell stabiliser element from Def.~\ref{def:shell} using the rules in Eq.~\eqref{eq:rules}. The power vertex $u$ does not contribute a generator element but two nearest neighbour chain vertices do, therefore, the $\tilde{C}$ has Pauli operator $\mathds{1}_u$ on power vertex $u$. The same holds for the chain vertices $r_{(u,v_1)}$ and $r_{(u,v_2)}$ for $r=1,3,\dots,2d-1$. The chain vertices $r_{(u,v_1)}$ and $r_{(u,v_2)}$ for $r=2,4,\dots,2d$ contribute a generator element but no nearest neighbours do, therefore the Pauli operators in $\tilde{C}$ is $X$. By definition, all other vertices have congruent Pauli operators. The overall sign factor is $+1$.\end{proof}

\begin{figure}
    \centering
    \begin{tikzpicture}
\begin{scope}[every node/.style={circle,thick,draw}]

    \node[line width=0.75mm,label={[label distance=0.05cm]90:\small $2$}] (1) at (1.5,6) {$\phantom{X}$};
    \node (121) at (2.25,6) {};
    \node (122) at (2.75,6) {};
    \node (123) at (3.25,6) {};
    \node (124) at (3.75,6) {};
    \node[line width=0.75mm, label={[label distance=0.05cm]90:\small $3$}] (2) at (4.5,6) {$\phantom{X}$};

    \node[line width=0.75mm] (241) at (4.125,5.25) {};
    \node (242) at (3.875,4.75) {};
    \node[line width=0.75mm] (243) at (3.625,4.25) {};
    \node (244) at (3.375,3.75) {};

    \node (131) at (1.125,5.25) {};
    \node (132) at (0.875,4.75) {};
    \node (133) at (0.625,4.25) {};
    \node (134) at (0.375,3.75) {};
    
    \node[label={[label distance=0.05cm]180:\small $1$}] (3) at (0,3) {$\phantom{X}$};

    \node[line width=0.75mm] (141) at (1.875,5.25) {};
    \node (142) at (2.125,4.75) {};
    \node[line width=0.75mm] (143) at (2.375,4.25) {};
    \node (144) at (2.625,3.75) {};

    \node[line width=0.75mm,label={[label distance=0.05cm]180:\small $4$}] (4) at (3,3) {$\phantom{X}$};

    \node (451) at (3.75,3) {};
    \node[line width=0.75mm] (452) at (4.25,3) {};
    \node (453) at (4.75,3) {};
    \node[line width=0.75mm] (454) at (5.25,3) {};

    \node[line width=0.75mm, label={[label distance=0.05cm]90:\small $5$}] (5) at (6,3) {$\phantom{X}$};

    \draw[-] (1) -- (121); 
    \draw[-] (121) -- (122);
    \draw[-] (122) -- (123);
    \draw[-] (123) -- (124);
    \draw[-] (124) -- (2);
    
    \draw[-] (2) -- (241); 
    \draw[-] (241) -- (242);
    \draw[-] (242) -- (243);
    \draw[-] (243) -- (244);
    \draw[-] (244) -- (4);

    \draw[-] (1) -- (131);
    \draw[-] (131) -- (132);
    \draw[-] (132) -- (133);
    \draw[-] (133) -- (134);
    \draw[-] (134) -- (3);

    \draw[-] (1) -- (141);
    \draw[-] (141) -- (142);
    \draw[-] (142) -- (143);
    \draw[-] (143) -- (144);
    \draw[-] (144) -- (4);

    \draw[-] (4) -- (451);
    \draw[-] (451) -- (452);
    \draw[-] (452) -- (453);
    \draw[-] (453) -- (454);
    \draw[-] (454) -- (5);

    \end{scope}
    \end{tikzpicture}

    \caption{Example of an inflated graph for $d=2$ from a graph with $5$ (large) vertices. The fat nodes highlight the vertices that measure a non-trivial Pauli operator when measuring the inflated generator element $f_4$ from Def.\,\ref{def:infvstab}. Power vertex $4$ and all highlighted chain vertices measure Pauli $X$ while the power vertices $2,3,5$ measure Pauli $Z$.}
    \label{fig:inflated1}
\end{figure}
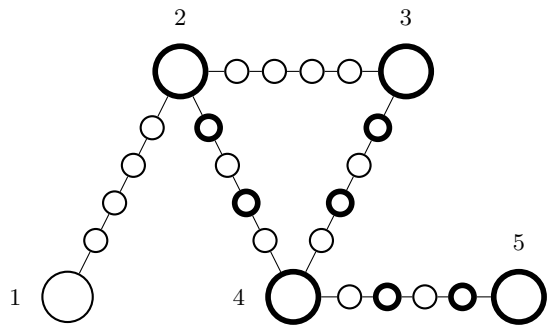

\section{Measurements inconsistent with $d$-LHV*-models}
\label{sec: main result}

Given our definitions, we can now move to finding the set of measurement which will allow us to apply lemma \ref{lem:comp}.
Here, we construct a set of Pauli measurements on the inflated graph state from the Pauli measurements on the original graph state. Then, we show that submeasurements of the new Pauli measurements fulfill Eqs.~\eqref{eq:infcond-3}, \eqref{ali:infcond-1}, \eqref{ali:infcond-2} 
on all vertices, i.e.\ they are proportional to stabiliser elements on the inflated graph state, their overall product equals $-\mathds{1}$ and all Pauli operators in the set of submeasurements, together with the Pauli operators from the Pauli measurements on the neighbours up to distance $d$, occur in even pairs. Thus, the new Pauli measurements exhibit correlations that no distance-$d$ communication-assisted LHV model can predict.

Given a set of Pauli measurements $(M_1,\dots,M_N)$ that fulfill Eqs.~\eqref{ali:cond-3}--\eqref{ali:cond-2} on the original graph state, we construct a set of Inflated measurements $(M^\prime_1,\allowbreak\dots,\allowbreak M^\prime_N)$ according to Def.~\ref{def:powerm}. Furthermore, we define Inflated stabiliser elements $(S^\prime_1,\dots,S^\prime_N)$ following Def.~\ref{def:powers} for the stabiliser elements $(S_1 = \chi^{\phantom{\prime}}_1 M_1,\dots,S_N = \chi^{\phantom{\prime}}_N M_N)$ on the original graph state from Eq.~\eqref{ali:cond-1}.

\begin{lemma} \label{lem:powers} Starting from a Pauli measurement $M_k=\chi^{\phantom{\prime}}_k S_k$ on the original graph state, given the respective Inflated measurement $M^\prime_k$ from Def.~\ref{def:powerm} and Inflated stabiliser element $S^\prime_k$ from Def.~\ref{def:powers}, then $C^\prime_k = \chi^\prime_k S^\prime_k$ is a submeasurement of the Inflated measurement $M_k^\prime$ and $\chi^\prime_k=\chi^{\phantom{\prime}}_k$. \end{lemma}

\begin{proof} First, we map the Inflated stabiliser element from its decomposition in generator elements to the representation with a sign factor $\chi^\prime_k$ and local Pauli operators using Eq.~\eqref{eq:rules}. From Def.~\ref{def:powers} the power vertices contribute a generator element if and only if a vertex on the original graph. The nearest-neighbour (a chain vertex) of a power vertex on the inflated graph contributes a generator element if and only if a nearest-neighbour vertex on the original graph contributes a generator element. Thus, the contribution of generator elements on the power vertices is the same as for the original vertices. As a result, a power vertex has the same Pauli operator with the same local sign as the vertex on the original stabiliser element. Therefore, the submeasurement perfectly coincides with the Pauli measurements on the power vertices and the overall sign is the same.\\
For the chain vertices, they contribute a generator element if and only if their neighbors at distance $2$. As a result the respective $t_w \, \text{mod} \, 2 = 0$ from Eq.\,\eqref{eq:rules} and therefore, $\sigma_w = X_w$ or $\sigma_w = \mathds{1}_w$, depending on whether chain vertex $w$ contributes a generator element. This complies with the Inflated measurements where all Pauli operators on the chain vertices are $X$. A minus sign occurs if chain vertex $w$ and its nearest neighbors contribute a generator element, in which case all chain vertices of a given chain do. Since there are an even number of chain vertices in every chain the cumulative sign is $+1$. Altogether the chain vertices do not invoke a sign change.\end{proof}
%
\begin{lemma}\label{lem:fulf} The submeasurements $(C^\prime_1=\chi^{\phantom{\prime}}_1 S^\prime_1,
\allowbreak \dots, \allowbreak C^\prime_N=\chi^{\phantom{\prime}}_N S^\prime_N)$ of the Inflated measurements $(M^\prime_1,\allowbreak\dots,\allowbreak M^\prime_N)$ fulfill Eq.~\eqref{ali:cond-3} for all vertices and Eq.~\eqref{ali:cond-1}, \eqref{ali:cond-2} on the inflated graph state.\end{lemma}

\begin{proof} According to Lemma \ref{lem:powers} the submeasurements fulfill Eq.~\eqref{ali:cond-1}. In its proof, we also show that $(C^\prime_k)_u = (M_k)_u$ on all power vertices $u$. Furthermore, the chain vertices measure Pauli operator $X$ if and only if they contribute a generator element, which they do if and and only if the closest power vertex at even distance (from Def.\,\ref{def:infvstab}). Due to,\,\eqref{ali:cond-3} for the original Pauli measurements, the Pauli operators and generator elements occur an even number of times across the set of original Pauli measurements (which are proportional to stabiliser elements on the original graph). As a result, Pauli operators occur an even number of times in the set of submeasurements for all power and chain vertices. Consequently, Eq.~\eqref{ali:cond-3} holds for all vertices. Lastly, following Lemma \ref{lem:powers}, it is $\chi^{\phantom{\prime}}_k = \chi^{\prime}_k $ and therefore Eq.~\eqref{ali:cond-2} holds for the set of submeasurements.\end{proof}

\begin{lemma}\label{lem:robp}The submeasurements 
$(C^\prime_1=\chi^{\phantom{\prime}}_1 S^\prime_1,\allowbreak \dots, \allowbreak C^\prime_N=\chi^{\phantom{\prime}}_N S^\prime_N)$ and Inflated measurements $(M^\prime_1,\dots,M^\prime_N)$ fulfill Eq.~\eqref{eq:infcond-3} on the power vertices.\end{lemma}

\begin{proof}The Inflated measurements have Pauli operator $X$ on all chain vertices. Up to distance $d$, all neighbour vertices of any power vertex $u$ are chain vertices, thus $(\mathbf{M}_k)_u^d= \sigma^k_u \bigotimes_{0 < \vert u-v \vert\leq d} X_v$. Since the $\lbrace \sigma^k_u \rbrace_k = \lbrace (M_k)_u \rbrace_k$ fulfill Eq.~\eqref{ali:cond-3} for all $u \in V$, the $\lbrace M^\prime_k, C^\prime_k \rbrace_k$ fulfill Eq.~\eqref{eq:infcond-3} on all power vertices.\end{proof}

\begin{lemma} \label{lem:robc}
Either the submeasurements $(C^\prime_1=\chi^{\phantom{\prime}}_1 S^\prime_1,\allowbreak \dots,\allowbreak C^\prime_N=\chi^{\phantom{\prime}}_N S^\prime_N)$ and Inflated measurements $(M^\prime_1, \allowbreak \dots,\allowbreak M^\prime_N)$ fulfill Eq.~\eqref{eq:infcond-3} on the chain vertices or we can add pairs of decoy measurements from Def.~\ref{def:decoy} with a Shell stabiliser element as both submeasurements to the set of Pauli measurements such that the complete set fulfills Eq.~\eqref{eq:infcond-3} on all vertices.
\end{lemma}

\begin{proof}
We evaluate Eq.\,\eqref{eq:infcond-3} for a given chain vertex $w$. 
Let us call $u$ the power vertex at the end of $w$'s chain with even distance, then $(C^\prime_k)_w\neq \mathds{1}_w$ if and only if $u$ contributes a generator element to all $S^\prime_k$ with $C^\prime_k = \chi_k S^\prime_k$. 
Up to distance $d$ along the graph's edges, almost all neighbours of $w$ are chain vertices, thus measuring Pauli operator $X$, the only exception being the closest 
power vertex of $w$, which we call $v$ ($\vert v -w \vert \leq d$).
Therefore, if $u$ contributes a generator element to $S^\prime_k$, then we are interested in $(\mathbf{M}_k)^d_w = \sigma_v \bigotimes_{w^\prime \neq v, 0 \leq \vert w^\prime -w \vert \leq d } X_{w^\prime} $.

For $u = v$, that is for $|w-v|$ even, Eq.\,\eqref{eq:infcond-3} holds for chain vertex $w$, since $(C^\prime_k )_u = (M^\prime_k )_u $ and Lemma \ref{lem:robp} holds on all power vertices $u$.

For $u \neq v$, that is if $|w-v|$ odd, we study how Eq.\,\eqref{eq:infcond-3} might not hold on vertex $w$. Consider a set of four Inflated measurements $M^\prime_{k_1},M^\prime_{k_2},M^\prime_{l_1},M^\prime_{l_2}$ with submeasurements, such that $(C^\prime_{k_1})_u = (C^\prime_{l_1})_u$, $(C^\prime_{k_2})_u = (C^\prime_{l_2})_u$, $(C^\prime_{k_1})_u \neq (C^\prime_{k_2})_u$. Furthermore, recall $C^\prime = \chi_k S^\prime$, power vertex $v$ contributes a generator element to both $S^\prime_{k_{1,2}}$ while it does to neither of $S^\prime_{l_{1,2}}$. Since Eq.\,\eqref{eq:infcond-3} holds on both power vertices $u,v$, such a set either exists or Eq.\,\eqref{eq:infcond-3} holds on vertex $w$.

Since the submeasurement are proportional to Inflated Stabiliser elements, the Pauli operators on the power vertex $u$ depends on the contribution of generator elements from main vertex $u$ and its power vertex neighbors, including $v$, following Eq.\,\eqref{eq:rules}. For two Inflated stabiliser elements with the same Pauli operator on vertex $u$, the contribution to the generator elements from $u$ is the same and from its neighbors modulo two. In the quartet, the measurements $k_1,l_1$ have the same Pauli operator on power vertex $u$ but the power vertex neighbor $v$ does not contribute a generator element to the stabiliser elements they are proportional to. Therefore, there must exist a second power vertex $v^\prime$, which is a power vertex neighbor of $u$ (not necessarily $v$), with opposing contribution of generator element to the stabiliser elements $k_1,l_1$. The same holds for $k_2,l_2$. Note that there might be additional power vertices, but since they must occur in pairs, they can also make their own set of four Inflated measurements. Note furthermore, that, while we focussed on chain vertex $w$, the above quartet accounts of all chain vertices with nearest power vertex $u$ and nearest power vertex at even distance $v$ or $v^\prime$, for all of which Eq.\,\eqref{eq:infcond-3} does not hold.

For such a quartet of Inflated Pauli measurements $M^\prime_{k_1},M^\prime_{k_2},M^\prime_{l_1}, M^\prime_{l_2}$ we add two decoy measurements $\tilde{M}^\prime_1$ and $\tilde{M}^\prime_2$ from Def.~\ref{def:decoy} with one Shell stabilizer element from Def.~\ref{def:shell} for both submeasurements around the power vertex $u$ with neighbor power vertices $v$ and $v^\prime$. The Pauli operators are $(\tilde{M}^\prime_1)_u = (M^\prime_{k_1})_u = (M^\prime_{l_1})_u$ and $(\tilde{M}^\prime_2)_u = (M^\prime_{k_2})_u = (M^\prime_{l_2})_u$. It might occur, that two of the four measurements already have the form of Shell stabiliser elements, in which case, we do not need to add the decoy measurements but alter the two measurements on the power vertex $u$ to have the same Pauli operator as the other to measurements of the quartet.

The added decoy measurements alter Eq.~\eqref{eq:infcond-3} for chain vertices whose closest power vertex is $u$. Additionally, their submeasurement only measure Pauli operator $X$ on chain vertices with odd distance to power vertex $u$. As a result, they target precisely those chain vertices, for which Eq.~\eqref{eq:infcond-3} did not hold. One decoy measurement corrects Eq.~\eqref{eq:infcond-3} either for two chain branches and the same Pauli operator $\sigma_u$ or for the same chain branch and different Pauli operators.\end{proof}

Altogether, the set of Inflated measurements and submeasurements constructed in Def.\ref{def:powerm} and Def.\ref{def:powers} from the original set of Pauli measurements fulfill the Eqs.~\eqref{ali:infcond-1}--\eqref{eq:infcond-3} on all power vertices and most chain vertices. If the above does not hold, one can always isolate responsible sets of four Inflated measurements. Then Def.\ \ref{def:decoy} provides an algorithm to construct specific pairs of decoy measurements and submeasurements such that the overall set fulfills Eq.~\eqref{eq:infcond-3} on all vertices. The additional submeasurements do not impede the results of Lemma \ref{lem:fulf} as they occur in equal pairs.

\section{Bell Inequality} \label{sec: bell inquality}

It is often useful to translate statements of non-locality into statistically robust inequalities, in order to experimentally test, or indeed use these kind of statements.
Given a GHSZ-like paradox on a graph state, i.e. a set of measurements that fulfills Eqs.\ \eqref{ali:cond-3}--\eqref{ali:cond-2}, one can sum up the measurements' operators to form a Bell operator and obtain a corresponding Bell inequality for any LHV model. Specifically, take Eq.\,\eqref{ali:cond-1} and define \( \mathcal{B} = \sum_k \chi_k M_k = \sum_k S_k \)
Suppose the GHSZ-like paradox consists of $N$ measurements, any LHV model predicts 
\begin{equation} \langle \mathcal{B} \rangle_{\text{LHV}} = \sum_{k=1}^N \chi_k < N = \langle \mathcal{B} \rangle_{\text{QM}}  \,. \label{eq:bellin} \end{equation} Due to Eq.\,\eqref{ali:cond-2} the inequality must be strict. 

One can do the same thing for the inflated graphs, to give bell inequalities for $d$-LHV$^\ast$ models. We can define a Bell operator from the set of measurements and corresponding submeasurements on the $d$-inflated graph state, that is $\mathcal{B}^\prime = \sum_k \chi^\prime_k C^\prime_k$. The constituting $N^\prime$ submeasurements form a Bell inequality for any communication-assisted LHV ($d$-LHV$^\ast$) model, \begin{equation} \langle \mathcal{B}^\prime \rangle_{d\text{-LHV}^\ast} = \sum_{k=1}^{N^\prime} \chi^\prime_k < N^\prime = \langle \mathcal{B}^\prime \rangle_{\text{QM}} \,, \label{eq:bellinin} \end{equation} as they fulfill Eqs.\,\eqref{eq:infcond-3}--\eqref{ali:infcond-2}. 
The relations between these values and the number $2s$ of added decoy measurements from Def. \ref{def:shell} that sustain the paradox in the
$d$-LHV$^\ast$ setting are simple:
\begin{gather}
    \ev{\mathcal B'}_{\text{QM}}=N'=N+2s=\ev{\mathcal B}_{\text{QM}}+2s,\\
    \ev{\mathcal B'}_{d\text{-LHV}^\ast}= \sum_{k=1}^{N\mathrlap{'}} \chi^\prime_k= \sum_{k=1}^N \chi_k+2s=\ev{\mathcal B}_{\text{LHV}}+2s,
\end{gather}
The violation ratio of the Bell inequality for the original graph state, against any LHV model, is therefore greater than or equal to the violation ratio of the Bell inequality for the inflated graph state against any $d$-LHV$^\ast$ model: 
 \begin{equation*}
     \frac{\ev{\mathcal B}_{\text{QM}}}{\ev{\mathcal B}_{\text{LHV}}}
     \ge
     \frac{\ev{\mathcal B}_{\text{QM}}+2s}{\ev{\mathcal B}_{\text{LHV}}+2s} = 
     \frac{\ev{\mathcal B'}_{\text{QM}}}{\ev{\mathcal B}_{d\text{-LHV}^\ast}}.
 \end{equation*}
\section{Small Graph States}
\label{sec:example}
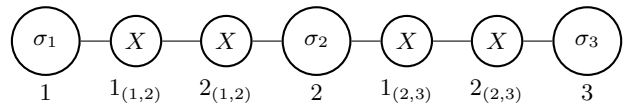
\begin{figure}
    \centering
\begin{tikzpicture}
\begin{scope}[main/.style={circle,thick,draw,minimum size=9mm},chain/.style={circle,thick,draw}]
    \node[main,label={[label distance=0.0cm]270:\small $1$}] (11) at (0,0) {$\sigma_1$};
    \node[chain,label={[label distance=0.05cm]270:\small $1_{(1,2)}$}] (21) at (1.2,0) {$X$};
    \node[chain,label={[label distance=0.05cm]270:\small $2_{(1,2)}$}] (31) at (2.4,0) {$X$};
    \node[main,label={[label distance=0.0cm]270:\small $2$}] (41) at (3.6,0) {$\sigma_2$};
    \node[chain,label={[label distance=0.05cm]270:\small $1_{(2,3)}$}] (51) at (4.8,0) {$X$};
    \node[chain,label={[label distance=0.05cm]270:\small $2_{(2,3)}$}] (61) at (6,0) {$X$};
    \node[main,label={[label distance=0.0cm]270:\small $3$}] (71) at (7.2,0) {$\sigma_3$};
    \draw[-] (11.east) -- (21.west);
    \draw[-] (21.east) -- (31.west);
    \draw[-] (31.east) -- (41.west);
    \draw[-] (41.east) -- (51.west);
    \draw[-] (51.east) -- (61.west);
    \draw[-] (61.east) -- (71.west);
\end{scope}
\end{tikzpicture}
    \caption{Inflated graph for $d=1$ with $7$ vertices in a chain from smallest non-trivial graph with $3$ vertices in a chain (considering only large balls). The symbols inside the balls denote Pauli operators of Pauli measurements on the corresponding graph states.}
    \label{fig:example}
\end{figure}
We now move to consider what are the smallest examples that can contradict $d$-LHV* models. In the original work of \cite{barrett2007modeling} the smallest example used 12 qubits. Our construction can use 7 starting from a line of three (See below). We will further find that through a slightly different construction we can show a contradition with $1$-LHV* models using linear graph state of 4 qubits. Interestingly, to have examples using graph states of fewer than 5 qubits, one requires non-Pauli measurements.

\subsection{Smallest 1-LHV*-violating inflated graph}

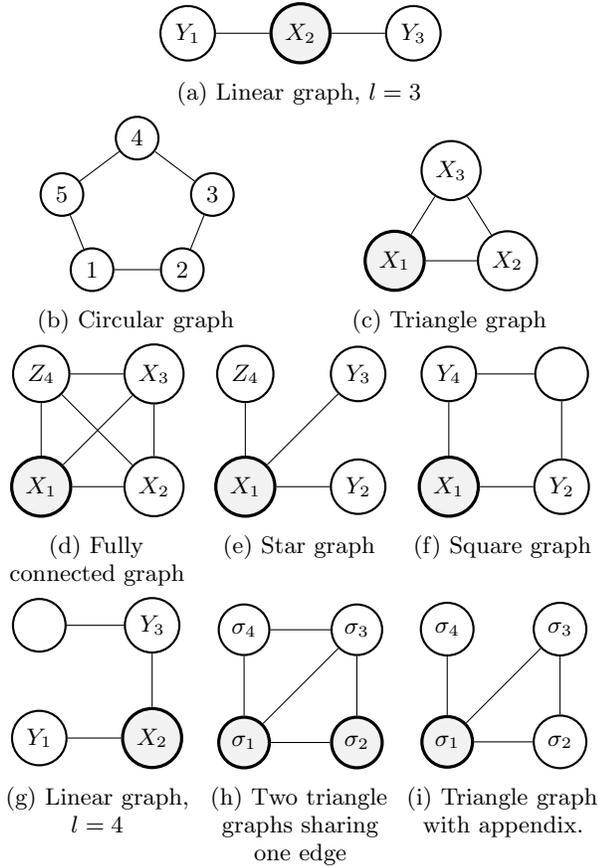
\begin{figure}
     \centering
\begin{subfigure}[b]{0.75\columnwidth}
\begin{tikzpicture}\begin{scope}[highlight/.style={circle, draw=black!100, fill=black!5, very thick, minimum size=5mm}, normal/.style={circle,thick,draw}]
    \node[normal] (1) at (0,0) {$Y_1$};
    \node[highlight] (2) at (1.5,0) {$X_2$};
    \node[normal] (3) at (3,0) {$Y_3$};
    \draw (1) -- (2);
    \draw (2) -- (3);
    \end{scope}\end{tikzpicture}
            \caption{Linear graph, $l=3$}
        \label{fig:ling3}
     \end{subfigure}
\begin{subfigure}[b]{0.45\columnwidth}
\begin{tikzpicture}
\begin{scope}[every node/.style={circle,thick,draw}]
    \node (1) at (0.9,0) {$1$};
    \node (2) at (2.1,0) {$2$};
    \node (3) at (2.5,1) {$3$};
    \node (4) at (1.5,1.75) {$4$};
    \node (5) at (0.5,1) {$5$};
    \draw (1) -- (2);
    \draw (1) -- (5);
    \draw (2) -- (3);
    \draw (3) -- (4);
    \draw (4) -- (5);
    \end{scope}
\end{tikzpicture}
\caption{Circular graph}
\label{fig:circulargraph}
\end{subfigure}
\begin{subfigure}[b]{0.49\columnwidth}
\begin{tikzpicture}\begin{scope}[highlight/.style={circle, draw=black!100, fill=black!5, very thick, minimum size=5mm}, normal/.style={circle,thick,draw}]
    \node[highlight] (1) at (0,0) {$X_1$};
    \node[normal] (2) at (1.5,0) {$X_2$};
    \node[normal] (3) at (0.75,1.2) {$X_3$};
    \draw (1) -- (2);
    \draw (1) -- (3);
    \draw (2) -- (3);
    \end{scope}\end{tikzpicture}
            \caption{Triangle graph}
        \label{fig:trig}
     \end{subfigure}
\begin{subfigure}[t]{0.3\columnwidth}
\begin{tikzpicture}\begin{scope}[highlight/.style={circle, draw=black!100, fill=black!5, very thick, minimum size=5mm}, normal/.style={circle,thick,draw}]
    \node[highlight] (1) at (0,0) {$X_1$};
    \node[normal] (2) at (1.5,0) {$X_2$};
    \node[normal] (4) at (0,1.5) {$Z_4$};
    \node[normal] (3) at (1.5,1.5) {$X_3$};
    \draw (1) -- (2);
    \draw (1) -- (3);
    \draw (1) -- (4);
    \draw (2) -- (3);
    \draw (2) -- (4);
    \draw (3) -- (4);
    \end{scope}\end{tikzpicture}
            \caption{Fully connected graph}
        \label{fig:fullyg}
     \end{subfigure}
\begin{subfigure}[t]{0.3\columnwidth}
\begin{tikzpicture}\begin{scope}[highlight/.style={circle, draw=black!100, fill=black!5, very thick, minimum size=5mm}, normal/.style={circle,thick,draw}]
    \node[highlight] (1) at (0,0) {$X_1$};
    \node[normal] (2) at (1.5,0) {$Y_2$};
    \node[normal] (4) at (0,1.5) {$Z_4$};
    \node[normal] (3) at (1.5,1.5) {$Y_3$};
    \draw (1) -- (2);
    \draw (1) -- (3);
    \draw (1) -- (4); \end{scope}\end{tikzpicture}
        \caption{Star graph}
        \label{fig:starg}
     \end{subfigure}
\begin{subfigure}[t]{0.3\columnwidth}
\begin{tikzpicture}\begin{scope}[highlight/.style={circle, draw=black!100, fill=black!5, very thick, minimum size=5mm}, normal/.style={circle,thick,draw,minimum size=7mm}]
    \node[highlight] (1) at (0,0) {$X_1$};
    \node[normal] (2) at (1.5,0) {$Y_2$};
    \node[normal] (4) at (0,1.5) {$Y_4$};
    \node[normal] (3) at (1.5,1.5) {};
    \draw (1) -- (2);
    \draw (1) -- (4);
    \draw (2) -- (3);
    \draw (3) -- (4); \end{scope}\end{tikzpicture}
        \caption{Square graph}
        \label{fig:squg}
     \end{subfigure}
     \begin{subfigure}[t]{0.3\columnwidth}
\begin{tikzpicture}\begin{scope}[highlight/.style={circle, draw=black!100, fill=black!5, very thick, minimum size=5mm}, normal/.style={circle,thick,draw, minimum size=7mm}]
    \node[normal] (1) at (0,0) {$Y_1$};
    \node[highlight] (2) at (1.5,0) {$X_2$};
    \node[normal] (4) at (0,1.5) {};
    \node[normal] (3) at (1.5,1.5) {$Y_3$};
    \draw (1) -- (2);
    \draw (2) -- (3);
    \draw (3) -- (4); \end{scope}\end{tikzpicture}
        \caption{Linear graph, $l=4$}
     \end{subfigure}
\begin{subfigure}[t]{0.3\columnwidth}
\begin{tikzpicture}\begin{scope}[highlight/.style={circle, draw=black!100, fill=black!5, very thick, minimum size=5mm}, normal/.style={circle,thick,draw}]
    \node[highlight] (1) at (0,0) {$\sigma_1$};
    \node[highlight] (2) at (1.5,0) {$\sigma_2$};
    \node[normal] (4) at (0,1.5) {$\sigma_4$};
    \node[normal] (3) at (1.5,1.5) {$\sigma_3$};
    \draw (1) -- (2);
    \draw (1) -- (3);
    \draw (1) -- (4);
    \draw (2) -- (3);
    \draw (3) -- (4); \end{scope}\end{tikzpicture}
        \caption{Two triangle graphs sharing one edge }
        \label{fig:twotg}
     \end{subfigure}
\begin{subfigure}[t]{0.3\columnwidth}
\begin{tikzpicture}\begin{scope}[highlight/.style={circle, draw=black!100, fill=black!5, very thick, minimum size=5mm}, normal/.style={circle,thick,draw,,minimum size=7mm}]
    \node[highlight] (1) at (0,0) {$\sigma_1$};
    \node[normal] (2) at (1.5,0) {$\sigma_2$};
    \node[normal] (4) at (0,1.5) {$\sigma_4$};
    \node[normal] (3) at (1.5,1.5) {$\sigma_3$};
    \draw (1) -- (2);
    \draw (1) -- (3);
    \draw (2) -- (3);
    \draw (1) -- (4); \end{scope}\end{tikzpicture}
        \caption{Triangle graph with appendix.}
        \label{fig:tappg}
     \end{subfigure}
        \caption{Circular Graph of $5$ vertices in b) and all connected graphs with $3$ and $4$ vertices. Highlighted vertices flip output variable's sign in $1$-LHV$^\ast$ model if they and their direct neighbours measure the indicated Pauli operators. Unspecified operators might be any Pauli operator. Symmetrically equivalent vertices and operators also flip the variables' sign. In (h) it is $\sigma_1 \sigma_2 \sigma_3 \sigma_4 =X_1Y_2Z_3Y_4$ for vertex 1 and $\sigma_1 \sigma_2 \sigma_3=Y_1X_2Y_3$ for vertex 2. In (i) it is either $\sigma_1 \sigma_2 \sigma_3 \sigma_4=Y_1X_2X_3Y_4$ or $\sigma_1 \sigma_2 \sigma_3 \sigma_4=X_1X_2X_3Z_4$.}
        \label{fig:signflip}
\end{figure}

Consider first the linear graph state with $7$ vertices drawn in Fig.\ \ref{fig:example}. It can be considered as an $d=1$-inflated graph of the linear graph state with three qubits in a line, which is locally equivalent to the GHZ state. Here, no LHV model can reproduce the measurement outcome of the four Pauli products $M_1 = Y_1 X_2 Y_3 $, $M_2 = Y_1 Y_2 Z_3$, $M_3 = Z_1 Y_2 Y_3$, $M_4 = Z_1 X_2 Z_3$. 
Then, we construct the Pauli measurements and submeasurements (omitting grayed out Pauli operators)
\begin{align*}
M^\prime_1 &= Y_1\, X_{1_{(1,2)}} X_{2_{(1,2)}} X_2\, X_{1_{(2,3)}} X_{2_{(2,3)}} Y_3\, \,,\\
M^\prime_2 &= Y_1\,X_{1_{(1,2)}} X_{2_{(1,2)}} \, Y_2\,\, \textcolor{gray}{X_{1_{(2,3)}}} X_{2_{(2,3)}} Z_3 \,,\\ 
M^\prime_3 &= Z_1 X_{1_{(1,2)}} \textcolor{gray}{X_{2_{(1,2)}}} \, Y_2\,\, X_{1_{(2,3)}} X_{2_{(2,3)}} Y_3\, \,,\\
M^\prime_4 &= Z_1 X_{1_{(1,2)}} \textcolor{gray}{X_{2_{(1,2)}}} X_2\, \textcolor{gray}{X_{1_{(2,3)}}} X_{2_{(2,3)}} Z_3 \,,\\
\tilde{M}^\prime_1 &= X_1 \textcolor{gray}{X_{1_{(1,2)}}} X_{2_{(1,2)}} \textcolor{gray}{X_2}\, X_{1_{(2,3)}} \textcolor{gray}{X_{2_{(2,3)}}} X_3 \,,\\
\tilde{M}^\prime_2 &=  X_1 \textcolor{gray}{X_{1_{(1,2)}}} X_{2_{(1,2)}} \, \textcolor{gray}{Y_2}\, X_{1_{(2,3)}} \textcolor{gray}{X_{2_{(2,3)}}} X_3 \,,
\end{align*}
which leads to the constraints  $\langle M^\prime_k \rangle_{\text{QM}} = \langle M^\prime_k \rangle_{1\text{-LHV}^\ast}$, for any nearest-neighbour communication-assisted LHV model
\begin{align*}
-1 &= y_1\, x^Y_{1_{(1,2)}} x^X_{2_{(1,2)}} x_2\, x^X_{1_{(2,3)}} x^Y_{2_{(2,3)}} y_3\, \,, \\
\phantom{-}1 &= y_1\, x^Y_{1_{(1,2)}} x^Y_{2_{(1,2)}} y_2\,\, \phantom{x_{1_{(2,3)}}^Y} x^Z_{2_{(2,3)}} z_3 \,,\\ 
\phantom{-}1 &= z_1 x^Z_{1_{(1,2)}} \phantom{x_{2_{(1,2)}}^Y} y_2\,\, x^Y_{1_{(2,3)}} x^Y_{2_{(2,3)}} y_3\, \,, \\
\phantom{-}1 &= z_1 x_{1_{(1,2)}}^Z \phantom{x_{2_{(1,2)}}^X} x_2 \, \phantom{x_{1_{(2,3)}}^X} x^Z_{2_{(2,3)}} z_3\,,\\
\phantom{-}1 &= x_1 \phantom{x_{1_{(1,2)}}^X} x^X_{2_{(1,2)}} \phantom{x_4}\, x^X_{1_{(2,3)}} \phantom{x_{2_{(2,3)}}^X} x_3\,,\\
\phantom{-}1 &= x_1 \phantom{x_{1_{(1,2)}}^X} x_{2_{(1,2)}}^Y \phantom{y_4}\,\, x^Y_{2_{(2,3)}} \phantom{x_{2_{(2,3)}}^X} x_3 \,,
\end{align*}
The superscripts denote the information acquired from communication with the vertices that change their Pauli operator in the different measurements. Multiplying the constraints results in $-1 = 1$, a contradiction. Thus, no 1-LHV$^\ast$ model can predict all correlations for the $7$ qubit chain graph state. Note that the Pauli operators for Pauli products $M^\prime_{1,2,3,4}$ match the ones in $M_{1,2,3,4}$ on the power vertices $1,2,3$. The Pauli operators $M^\prime_{5,6}$ are added as decoy measurements so that the LHV model fails to distinguish between power vertex $2$ measuring $X_2$ or $Y_2$ on the adjacent chain vertices.

More quantitatively, this example shows a 1-LHV$^\ast$ violation ratio of $\tfrac{6}{4}=\tfrac{3}{2}$ by 
by inflating the 3-qubit linear graph state to 7 qubits, the later showing a LHV violation ration of $\tfrac{4}{2}=2$.

\subsection{Smallest 1-LHV*-violating graph through Pauli measurements}

The smallest graph state where no 1-LHV*
model can reproduce submeasurements from Pauli measurements is the one with $5$ vertices in a circle which is drawn in Fig.\,\ref{fig:circulargraph}. These do not follow the structure of the Inflated Pauli measurements and $16$ Pauli measurements are needed that exploit the rotational symmetry of the cyclic graph. These are
\begin{align*}
&M^\prime_{1} \hspace{-1.5cm}&&= \,\,X_1 X_2 X_3 X_4 X_5 \,\,\,, \\
&M^\prime_{2+k} \hspace{-1.5cm}&&= (X_1 \textcolor{gray}{X_2} Y_3\, Y_4\, \textcolor{gray}{X_5})_k\,,\\
&M^\prime_{7+k} \hspace{-1.5cm}&&= (Y_1\,X_2 X_3 Y_4\, \textcolor{gray}{Y_5}\,)_k\,, \\
&M^\prime_{12+k} \hspace{-1.5cm}&&= (Y_1\,X_2 X_3 Y_4\, \textcolor{gray}{X_5})_k\,,
\end{align*}
with $k=0,\dots,4$ denoting a cyclic rotation of the operators. Comparing $\langle M^\prime_k \rangle_{\text{QM}}  = \langle M^\prime_k \rangle_{1-\text{LHV}^\ast}$ leads to constraints on the variables
\begin{align*}
-1 &= x^{XX}_1 x^{XX}_2 x^{XX}_3 x^{XX}_4 x^{XX}_5 \,, \\
\phantom{-}1 &= (x^{XX}_1 \phantom{x^{XY}_2} y^{XY}_3 y^{YX}_4\, \phantom{x^{YX}_5})_k \,,\\
\phantom{-} 1 &= (y^{YX}_1\, x^{YX}_2 x^{XY}_3 y^{XY}_4\, \phantom{y^{YY}_5}\, )_k\,,\\
\phantom{-} 1 &= (y^{XX}_1\,x^{YX}_2 x^{XY}_3 y^{XX}_4\, \phantom{x^{YY}_5})_k \,,
\end{align*}
which leads to a contradiction.
We have thus obtained a 1-LHV$^*$ violation of $\tfrac{16}{14}$ using Pauli measurements on five qubits.

In  order to prove that a graph state with $5$ qubits is the smallest graph state such that no 1-LHV$^*$ model can predict all Pauli measurements, we show that for every graph state with less than $5$ qubits, there exist a 1-LHV$^*$ model that correctly predicts the outcome of all Pauli measurements. The model is similar to one proposed by Barrett, Caves et al.\ in \cite{barrett2007modeling}.  Specifically, for every vertex $u$, we define hidden variables $z_u,x_u,y_u=\pm 1$ as the output values for a $Z_u,X_u,Y_u$ Pauli measurement. The variables' value $z_u \pm =1$ is chosen uniformly at random, then one calculates the values of variables $x_u = \prod_{(u,v)\in E} z_v$ and $y_u= x_u z_u$. If no measurement is performed on a vertex, the assumed output is $1$. So far the model is a LHV model since it does not rely on any classical communication. Every vertex has the variables $1,z_u,x_u,y_u$ at disposal as output given a local Pauli operator. Any Pauli measurements outcome is then a product of these variables.

We argue that this model correctly predicts any Pauli measurement except when $M= \bigotimes_u \sigma_u = - S$  for a stabiliser element in  Eq.\eqref{eq:stab-decomposition2}. Compare $1 = x_u \prod_{(u,v)\in E} z_v$ to Eq.\eqref{eq:vertexstabelement} to note that the model predicts the measurement of a generator element of the graph state correctly. Furthermore, $y_u = x_u z_u$ mimics the product of Pauli operators up to the missing imaginary unit $\mathrm{i}$. Therefore, the model only makes an error, if the missing imaginary units inflict a sign change, which happens exactly for the stabiliser elements with $\chi = -1$ in Eq.\,\eqref{eq:stab-decomposition2}. A more detailed reasoning can be found in \cite{barrett2007modeling}.

Including a round of classical communication between along the graph's edges up to distance $d$, a successful model selects the cases when a sign flip of a vertex's output variable is necessary, while remaining foolproof to a submeasurement scheme ignoring certain vertices' outputs. For every connected graph with less than $5$ vertices, we display the cases of a needed sign flip in Fig. \ref{fig:signflip} based on the LHV model described above. An analysis of all the possible Pauli measurements and submeasurements on the graph states reveals that the Distance-$1$-communication-assisted LHV model predicts the outcomes in accordance with a quantum measurement on the graph state.\\

\subsection{Smallest 1-LHV*-violating graph through Clifford measurements}

The smallest example for any violation of a 1-LHV*
model using binary inputs and outputs
is the distance-$1$ inflated CHSH inequality that resorts to Clifford measurements on a chain of $4$ qubits. The Bell operator is 
\begin{align}
    \mathcal{B} = &\left(R_X(\tfrac{\pi}{2})_1 X_2 + R_X(\tfrac{3\pi}{2})_1 X_2 \right) \textcolor{gray}{X_3} X_4 \nonumber \\ &+ \left(R_X(\tfrac{\pi}{2})_1 X_2 - R_X(\tfrac{3\pi}{2})_1 X_2 \right) X_3 Y_4 \,,
\end{align}
with $R_X(\theta)_1 = \cos(\tfrac{\theta}{2}) Z_1 + \sin(\tfrac{\theta}{2}) Y_1$ and the grayed-out operator output is ignored in the measurement's outcome. A 1-LHV$^*$ 
model obeys $\langle \mathcal{B} \rangle_{1\text{-LHV}^*} \leq 2$ while the measurement outcome on the $4$ qubit chain graph state equals $\langle \mathcal{B}_{\text{QM}} \rangle = 2\sqrt{2}$, leading to a violation ratio of $\sqrt2$. This can easily be checked using the same reasoning as above.


This 4-qubit graph is the smallest graph state showing violations for the restricted case of binary inputs and outputs.
Indeed, the both 3-qubit graphs states can be simulated classically in this setting.
This is obvious for the triangle, where everyone knows all the settings ; 
and Chaves et al.\ \cite[A.3]{Chaves2017Causal} 
showed through convex optimization
the class corresponding to a 3-qubit linear graph-state  in 1-LHV* 
--- the class $\qty{(1,3), (2,3), (1,2,3)}$ in their terminology --- 
is nonsignalling boring for binary inputs and outputs.

\section{Conclusion} \label{sec:conc}

From any graph states with more than three connected vertices we constructed a set of Pauli measurements on a $d$-inflated version of the original graph. The set contradicts any LHV model even assisted by classical communication along the graph's edges up to distance $d$. We state the violation in terms of a GHSZ-like paradox which we can write it in terms of a Bell inequality. The construction and proof works for any integer $d$, such that successive inflation is also possible.

Furthermore, we found a set of Pauli measurements for the smallest graph state that is robust against a LHV description including nearest-neighbour classical communication, the circular graph on $5$ qubits. The smallest graph state that defies a nearest-neighbour Communication-Assisted LHV model using binary inputs and outputs is a linear graph on $4$ qubits; 
we provide a Bell inequality that uses Clifford operators for this case.

These results are useful in two directions at least.
Firstly, we have found much smaller examples than previously, which are minimal in size using only four qubits, well within the reach of many experimental labs.
Secondly, the existence of correlations which cannot be explained by communication assisted LHV has already found applications in proving classical and quantum separation for shallow circuits \cite{bravyi2018quantum,bravyi2020quantum}, distributed computing \cite{gall2018quantum}, and giving novel ways to certify randomness \cite{coudron2018trading}. Given the extent that graph states are used in quantum information, from fault tolerance to quantum sensing, this work gives a good route to finding more applications or improving on those existing.

\bibliography{generic}

\end{document}